\documentclass{aims}
\usepackage{amsmath,amssymb}
  \usepackage{paralist}
  \usepackage{graphics} 
  \usepackage{epsfig} 
\usepackage{graphicx}  \usepackage{epstopdf}
\usepackage{subfigure}
\usepackage{overpic}
 \usepackage[colorlinks=true]{hyperref}
\hypersetup{urlcolor=blue, citecolor=red}

\def\currentvolume{X}
 \def\currentissue{X}
  \def\currentyear{20XX}
   \def\currentmonth{XX}
    \def\ppages{X--XX}
     \def\DOI{10.3934/xx.xx.xx.xx}
     
\newtheorem{theorem}{Theorem}[section]
\newtheorem{corollary}{Corollary}
\newtheorem{lemma}[theorem]{Lemma}
\theoremstyle{definition}
\newtheorem{definition}[theorem]{Definition}

\newcommand\bs{\boldsymbol{s}}
\newcommand\bxi{\boldsymbol{\xi}}
\newcommand\bx{\boldsymbol{x}}
\newcommand\be{\boldsymbol{e}}
\newcommand\bu{\boldsymbol{u}}
\newcommand\bw{\boldsymbol{w}}
\newcommand\bC{\boldsymbol{C}}
\newcommand\bQ{\boldsymbol{Q}}
\newcommand\bq{\boldsymbol{q}}
\newcommand\bTheta{\boldsymbol{\Theta}}
\newcommand\bphi{\boldsymbol{\phi}}
\newcommand\dd{\,\mathrm{d}}
\newcommand\bn{\boldsymbol{n}}

\newcommand\rd{\mathrm{d}}
\newcommand\odd[1]{\dfrac{\rd}{\rd {#1}}}

\newcommand\pdd[1]{\dfrac{\partial}{\partial {#1}}}
\newcommand\pd[2]{\dfrac{\partial {#1}}{\partial {#2}}}

\title[HYPERBOLICITY OF 13-MOMENT SYSTEMS]
      {On Hyperbolicity of 13-Moment System}

\author[Zhenning Cai and Yuwei Fan and Ruo Li]{}


 \email{caizn@pku.edu.cn}
 \email{ywfan@pku.edu.cn}
 \email{rli@math.pku.edu.cn}

\thanks{The research is supported in part by the National Basic Research
Program of China (2011CB309704)}

\begin{document}
\maketitle

\centerline{\scshape Zhenning Cai}
\medskip
{\footnotesize
 \centerline{CAPT \& School of Mathematical Sciences}
   \centerline{Peking University}
   \centerline{Yiheyuan Road 5, 100871 Beijing, China}
} 

\medskip

\centerline{\scshape Yuwei Fan}
\medskip
{\footnotesize
 \centerline{School of Mathematical Sciences}
   \centerline{Peking University}
   \centerline{Yiheyuan Road 5, 100871 Beijing, China}
} 

\medskip

\centerline{\scshape Ruo Li}
\medskip
{\footnotesize
 \centerline{CAPT \& School of Mathematical Sciences}
   \centerline{Peking University}
   \centerline{Yiheyuan Road 5, 100871 Beijing, China}
}

\bigskip

 \centerline{(Communicated by Yan Guo)}

\begin{abstract}
  We point out that the thermodynamic equilibrium is not an interior
  point of the hyperbolicity region of Grad's 13-moment system. With a
  compact expansion of the phase density, which is compacter than
  Grad's expansion, we derived a modified 13-moment system. The new
  13-moment system admits the thermodynamic equilibrium as an interior
  point of its hyperbolicity region. We deduce a concise criterion to
  ensure the hyperbolicity, thus the hyperbolicity region can be
  quantitatively depicted.
\end{abstract}

\section{Introduction}
Grad's 13-moment system \cite{Grad} has been studied for over 50
years. This system was derived by utilizing the isotropic Hermite
expansion \cite{Grad1949note} of the phase density in the Boltzmann
equation, and such an idea opened a brand new direction in the gas
kinetic theory. In the subsequent years, a number of defects of this
model were discovered, one of which was that Grad's 13-moment system
is not globally hyperbolic, and for 1D flows, the hyperbolicity can
only be obtained near the equilibrium \cite{Muller}. The loss of
hyperbolicity directly breaks the well-posedness of the partial
differential equations, and thus the capability of this model is
strictly limited. Extensions of this model include the regularized
Burnett equations \cite{Jin}, regularized 13-moment equations
\cite{Struchtrup2003, Struchtrup2005}, and the Pearson-13-moment
equations \cite{Torrilhon2010}. These methods may alleviate the
problem of hyperbolicity to some extent \cite{Slemrod, TorrilhonRGD,
  Torrilhon2010}. However, all the analysis on the hyperbolicity is
restricted to the 1D flows, while there are no comments indicating
that the 3D case is the same or similar as the 1D case. Actually,
Grad's paper \cite{Grad} has pointed out that in comparison with the
1D flows, an additional soundspeed appears in the 2D flows. In this
paper, we are concerned with the hyperbolicity for the full 3D flows.

We first point out that in the 3D Grad's 13-moment equations, for each
equilibrium state, none of its neighbourhoods is contained in the
hyperbolicity region. This reveals that the manifold formed by all the
equilibrium states is on the boundary of the hyperbolicity region, and
thus an arbitrary small perturbation of the equilibrium state may lead
to the loss of hyperbolicity, which indicates the instability of
Grad's 13-moment equations. More precisely, we prove that if an
arbitrary small anisotropic perturbation is applied to the phase
density from the equilibrium, the hyperbolicity may be broken down.
Noticing that the anisotropy plays an essential role in breaking down
the hyperbolicity, we then propose a new modified 13-moment model such
that the equilibrium state lies in the interior of the hyperbolicity
region. This modified system is derived by an anisotropic Hermite
expansion instead of the isotropic Hermite expansion in Grad's method,
where the anisotropy is specified by the temperature tensor. We find
out a dimensionless quantity which can prescribe the departure of the
phase density from the equilibrium state. It is proven that if this
quantity is controlled above by a threshold, the full 3D system is
hyperbolic. The value of this threshold is approximately given, and
the size of the hyperbolicity region is depicted using the
Chapman-Enskog type expansion.

The rest of this paper is as follows: in Section \ref{sec:algebra},
some algebraic lemmas are given as preliminaries; in Section
\ref{sec:Boltzmann_Grad}, the hyperbolicity of Grad's 13-moment
equations is discussed for both 1D and 3D flows; Section
\ref{sec:new_13m} gives a modified 13-moment system, and its
hyperbolicity is discussed in detail; finally, some concluding remarks
are given in \ref{sec:conclusion}.



\section{Some preliminary results in linear algebra} \label{sec:algebra}
In this paper, we will focus on the hyperbolicity of moment systems,
i.e. the real diagonalizability of the coefficient matrices in these
systems. In order to make our later derivation self-consistent, we
present some lemmas about matrices and polynomials as preliminaries
and most of the proofs can be found in textbooks of linear algebra. If
not specified, we are considering real matrices and polynomials with
real coefficients.
\begin{lemma}[Cayley-Hamilton]
  For any square matrix $\bf A$ and its characteristic polynomial
  $p(\lambda) = \mathrm{det}(\lambda {\bf I} - {\bf A})$, we have
  $p({\bf A}) = {\bf 0}$.
\end{lemma}

\begin{lemma} \label{thm:mp_cp}
For a square matrix $\bf A$, the following three statements are
equivalent:
\begin{enumerate}
\item $\lambda$ is a root of the minimal polynomial of $\bf A$,
\item $\lambda$ is a root of the characteristic polynomial of $\bf A$,
\item $\lambda$ is an eigenvalue of $\bf A$.
\end{enumerate}
\end{lemma}

\begin{lemma} \label{thm:diag}
A square matrix $\bf A$ is diagonalizable if and only if the minimal
polynomial of $\bf A$ is the product of distinct linear functions.
\end{lemma}

The following corollary can be derived from the above results:
\begin{corollary} \label{prop:non_diag}
For a square matrix $\bf A$, suppose $p(\lambda)$ is its
characteristic polynomial. If there exists a polynomial $q(\lambda)$
such that $p(\lambda)$ and $q(\lambda)$ share the same roots, and
$q(\bf A) \neq {\bf 0}$, then $\bf A$ is not diagonalizable.
\end{corollary}
\begin{proof}
  Let us prove it by contradiction. We suppose $\bf A$ is
  diagonalizable.  According to Lemma \ref{thm:mp_cp} and Lemma
  \ref{thm:diag}, the minimal polynomial of $\bf A$ can be written as
\begin{equation}
m(\lambda) = (\lambda - \lambda_1) \cdots (\lambda - \lambda_n),
\end{equation}
where $\lambda_1, \cdots, \lambda_n$ are all the distinct eigenvalues
of $\bf A$. Since $p(\lambda)$ and $q(\lambda)$ share the same roots,
we have $m(\lambda) \mid q(\lambda)$, and thus $q({\bf A}) = {\bf 0}$,
which violates the condition $q({\bf A}) \neq {\bf 0}$.
\end{proof}
We will use Lemma \ref{thm:diag} when proving a matrix is
diagonalizable, and use Corollary \ref{prop:non_diag} when proving a
matrix is not diagonalizable. The following lemma can be used to tell
if a polynomial has multiple roots.

\begin{lemma} \label{lem:multi_root}
Let $p(z)$ be a polynomial defined on complex numbers. Then $p(z)$ has
no multiple roots if and only if $p(z)$ and $\odd{z} p(z)$ are
coprime.
\end{lemma}

\begin{definition}
Let $p(z)$ and $q(z)$ be two polynomials defined as
\begin{equation}
p(z) = p_0 + p_1 z + p_2 z^2 + \cdots + p_m z^m, \quad
q(z) = q_0 + q_1 z + q_2 z^2 + \cdots + q_n z^n.
\end{equation}
The $(m+n) \times (m+n)$ matrix
\begin{equation}
\mathrm{Syl}(p,q) :=
\begin{array}{l@{}l}
  \begin{pmatrix}
    p_m & \cdots & \cdots & p_0 & \\
    & \ddots & & & \ddots & \\
    & & p_m & \cdots & \cdots & p_0 \\
    q_n & \cdots & \cdots & q_0 & \\
    & \ddots & & & \ddots & \\
    & & q_n & \cdots & \cdots & q_0
  \end{pmatrix} 
&
  \begin{array}{l}
    \left. \vphantom{%
      \begin{matrix} p_0 \\ \ddots \\ p_0 \end{matrix}%
    }\right\}
    \text{\footnotesize $n$ rows} \\
    \left. \vphantom{%
      \begin{matrix} q_0 \\ \ddots \\ q_0 \end{matrix}%
    }\right\}
    \text{\footnotesize $m$ rows}
  \end{array}
\end{array}
\end{equation}
is called the Sylvester matrix associated to $p(z)$ and $q(z)$. The
resultant of $p(z)$ and $q(z)$ is defined as the determinant of the
above matrix:
\begin{equation}
\mathrm{res}(p,q) := \mathrm{det} \left( \mathrm{Syl}(p,q) \right).
\end{equation}
\end{definition}

The following lemma is to be used to tell if two polynomials have
common roots.
\begin{lemma} \label{lem:coprime}
Two nonzero polynomials are coprime if and only if their resultant is
nonzero.
\end{lemma}

\section{Boltzmann equation and Grad's 13-moment system}
\label{sec:Boltzmann_Grad}
From this section, we start our discussion on the kinetic models.
Section \ref{sec:Boltzmann} and \ref{sec:Grad13} give brief
introductions to the Boltzmann equation and Grad's 13-moment system
respectively, and Section \ref{sec:Grad13_1D} reviews the classical
knowledge on the hyperbolicity of Grad's 13-moment system for 1D
flows. In the last part of this section, we present our new findings
that an intrinsic difference exists between the hyperbolicity regions
in the 1D and 3D cases for Grad's system.

\subsection{Boltzmann equation} \label{sec:Boltzmann}
The Boltzmann equation is a fundamental physical model in the gas
kinetic theory. Suppose $f(t,\bx,\bxi)$ is the function of phase
density, where $t$ is the time, $\bx = (x_1, x_2, x_3)^T$ is the
spatial coordinates, and $\bxi = (\xi_1, \xi_2, \xi_3)^T$ stands for
the velocity of gas molecules. Then the Boltzmann equation reads
\begin{equation}
\pd{f}{t} + \bxi \cdot \nabla_{\bx} f
  = Q(f,f).
\end{equation}
The right hand side $Q(f,f)$ describes the interaction between
particles:
\begin{equation}
Q(f,f) = \int_{\mathbb{R}^3} \int_0^{2\pi} \int_0^{\pi/2}
  (f_*' f' - f_* f)
  |\bxi - \bxi_*| \sigma \sin \Theta
\dd \Theta \dd \varepsilon \dd \bxi_*,
\end{equation}
where $f' = f(t,\bx,\bxi')$, $f_* = f(t,\bx,\bxi_*)$, $f_*' =
f(t,\bx,\bxi_*')$, and the velocities $\bxi$, $\bxi_*$ and $\bxi'$,
$\bxi_*'$ are the pre- and post-collision velocities of a colliding
pair of molecules, and $\sigma$ is the differential cross-section. In
this paper, we consider only the Maxwell molecules, for which $|\bxi -
\bxi_*| \sigma$ is a function of $\Theta$. We refer the readers to
\cite{Cowling, Bird} for more details of the collision term and the
Maxwell molecules.

The gas kinetic theory describes the fluid states in a microscopic
view, while the macroscopic quantities such as density, velocity and
temperature can be obtained by integrations. Define
\begin{equation}
\langle h \rangle = m \int_{\mathbb{R}^3} h \dd \bxi,
\end{equation}
where $m$ is the mass of a single gas molecule. Then the relations
between the density function $f$ and some common macroscopic
quantities are as follows:
\begin{itemize}
\item Density: $\rho = \langle f \rangle$;
\item Velocity: $\bu = (u_1, u_2, u_3)^T =
  \rho^{-1} \langle \bxi f \rangle$;
\item Temperature: $T = (3 \rho k_B / m)^{-1} \langle |\bxi - \bu|^2 f
  \rangle$, where $k_B$ is the Boltzmann constant;
\item Temperature tensor: $T_{ij} = (\rho k_B / m)^{-1} \langle
  (\xi_i-u_i) (\xi_j-u_j) f \rangle$, $i,j = 1,2,3$;
\item Heat flux: $\bq = (q_1, q_2, q_3)^T = \langle |\bxi - \bu|^2
  (\bxi - \bu) f / 2 \rangle$.
\end{itemize}
Following the conventional style, we denote
\begin{equation}
\theta = \frac{k_B}{m} T, \quad \theta_{ij} = \frac{k_B}{m} T_{ij},
  \quad \bTheta = (\theta_{ij})_{3\times 3}.
\end{equation}
It can be derived from the positivity of the density function $f$ that
$\bTheta$ is symmetric positive definite. For simplicity, below we
denote the relative velocity $\bxi - \bu$ by $\bC$, and the norm of
$\bC$ is denoted by $C$. For example, we have
\begin{equation}
\theta = (3 \rho)^{-1} \langle C^2 f \rangle, \quad
\theta_{ij} = \rho^{-1} \langle C_i C_j f \rangle.
\end{equation}

\subsection{Grad's 13-moment system} \label{sec:Grad13}
The high dimensionality of the Boltzmann equation introduces extreme
difficulties to its numerical treatment. In order to simplify the
model, Grad proposed a 13-moment system \cite{Grad}, in which the
velocity variable $\bxi$ was eliminated, while only 13 equations
are presented. These equations are derived by assuming the following
particular form of the phase density $f$:
\begin{equation} \label{eq:dis_13m}
f|_{13} = \left[
  1 + \frac{\theta_{ij} - \delta_{ij} \theta}{2\theta^2} \left(
    C_i C_j - \delta_{ij} C^2
  \right) + \frac{2}{5} \frac{q_k}{\rho \theta^2} C_k \left(
    \frac{C^2}{2 \theta} - \frac{5}{2}
  \right)
\right] f_M,
\end{equation}
where $f_M$ is the Maxwellian, defined as
\begin{equation}
f_M = \frac{\rho}{(2\pi \theta)^{3/2}}
  \exp \left( -\frac{C^2}{2\theta} \right),
\end{equation}
and in \eqref{eq:dis_13m}, the Einstein summation convention is
assumed. Accordingly, when an index appears twice in a single term, it
implies summation of that term over all the values of the index. By
\eqref{eq:dis_13m}, Grad's 13-moment system can be written as
a closed system as
\begin{equation} \label{eq:G13_drv}
\left \langle \bphi \pd{f|_{13}}{t} \right \rangle
  + \langle \bphi (\bxi \cdot \nabla_{\bx} f|_{13}) \rangle
  = \langle \bphi Q(f|_{13},f|_{13}) \rangle,
\end{equation}
where 
\begin{displaymath}
\begin{split}
\bphi ={} & (1, \\
       & ~C_1, C_2, C_3, \\
       & ~C_1^2, C_2^2, C_3^2, C_1 C_2, C_1 C_3, C_2 C_3, \\
       & ~C^2 C_1, C^2 C_2, C^2 C_3)^T.
\end{split}
\end{displaymath}
By simplifications of the expression obtained after the integrations, the 
above system can be explicitly given by
\begin{equation} \label{eq:G13}
\begin{split}
& \frac{\mathrm{d} \rho}{\mathrm{d} t} +
  \rho \pd{u_k}{x_k} = 0, \\
& \frac{\mathrm{d} u_i}{\mathrm{d} t} +
  \frac{\theta_{ik}}{\rho} \pd{\rho}{x_k} +
  \pd{\theta_{ik}}{x_k} = 0, \quad i = 1,2,3, \\
& \frac{\mathrm{d} \theta_{ij}}{\mathrm{d} t} +
  2 \theta_{k(i} \pd{u_{j)}}{x_k} +
  \frac{1}{\rho} \left(
    \frac{4}{5} \pd{q_{(i}}{x_{j)}} +
    \frac{2}{5} \delta_{ij} \pd{q_k}{x_k}
  \right) =
    -\frac{\rho\theta}{\mu} (\theta_{ij} - \delta_{ij} \theta),
  \quad i,j = 1,2,3, \\
& \frac{\mathrm{d} q_i}{\mathrm{d} t} -
  (\theta_{ij}\theta_{jk} - 2\theta\theta_{ik} + \theta^2 \delta_{ik})
    \pd{\rho}{x_k} +
  \frac{7}{5} q_i \pd{u_k}{x_k} +
  \frac{7}{5} q_k \pd{u_i}{x_k} +
  \frac{2}{5} q_k \pd{u_k}{x_i} \\
& \qquad -\rho \theta_{ik} \left(
    \pd{\theta_{jk}}{x_j} -
    \frac{7}{6} \pd{\theta_{jj}}{x_k}
  \right) + 2\rho \theta \left(
    \pd{\theta_{ik}}{x_k} -
    \frac{1}{3} \pd{\theta_{jj}}{x_i}
  \right) = -\frac{2}{3} \frac{\rho \theta}{\mu} q_i, \quad i=1,2,3.
\end{split}
\end{equation}
Here, \[ \odd{t} = \pdd{t} + u_k \pdd{x_k} \]
is the material derivative, and the brackets around indices denote the
symmetrization of a tensor. The symbol $\mu$ denotes the coefficient
of viscosity. For Maxwell molecules, $\mu$ is proportional to
$\theta$.

\subsection{Local hyperbolicity of 1D Grad's moment system}
\label{sec:Grad13_1D}
It is well-known that the 1D Grad's moment system is hyperbolic only
when the fluid is near the thermodynamic equilibrium state. For 1D
flows, the 13-moment system reduces to a smaller system containing
only five equations, which are obtained by setting $u_2 = u_3 =
\theta_{12} = \theta_{13} = \theta_{23} = q_2 = q_3 = 0$ and
$\theta_{33} = \theta_{22}$ in \eqref{eq:G13}. Such operation
eliminates eight of the thirteen variables, and results in the 1D
system
\begin{equation} \label{eq:G13_1D}
\pd{\hat{\bw}}{t} +
  \hat{\bf M}(\hat{\bw}) \pd{\hat{\bw}}{x} =
  \hat{\bQ}(\hat{\bw}),
\end{equation}
where $\hat{\bw} = (\rho, u_1, \theta_{11}, \theta_{22}, q_1)^T$,
$\hat{\bQ}(\hat{\bw}) = \big( 0, 0, \rho \theta (\theta-\theta_{11})
/ \mu, \rho \theta (\theta-\theta_{22}) / \mu, -\frac{2}{3} \rho
\theta q_1 / \mu \big)^T$, and
\begin{equation}
\hat{\bf M}(\hat{\bw}) = \begin{pmatrix}
  u_1 & \rho & 0 & 0 & 0 \\
  \theta_{11} / \rho & u_1 & 1 & 0 & 0 \\
  0 & 2\theta_{11} & u_1 & 0 & 6/(5\rho) \\
  0 & 0 & 0 & u_1 & 2/(5\rho) \\
  -4(\theta_{11} - \theta_{22})^2 / 9 & 16q_1 / 5 &
    \rho(11\theta_{11} + 16\theta_{22}) / 18 &
    \rho(17\theta_{11} - 8\theta_{22}) / 9 & u_1
\end{pmatrix}.
\end{equation}
The system \eqref{eq:G13_1D} is hyperbolic if and only if $\hat{\bf M}
(\hat{\bw})$ is real diagonalizable.

In order to check the diagonalizability of $\hat{\bf M}(\hat{\bw})$,
we calculate its characteristic polynomial as
\begin{equation}
\begin{split}
\mathrm{det}(\lambda {\bf I} - \hat{\bf M}) &=
  (\lambda - u_1) \bigg[
    (\lambda - u_1)^4 -
    \frac{2}{45} (101 \theta_{11} + 16 \theta_{22}) (\lambda-u_1)^2 \\
& \qquad \qquad -\frac{96}{25} \frac{q_1}{\rho} (\lambda - u_1) +
    \frac{1}{15} (53\theta_{11}^2 - 16 \theta_{11} \theta_{22} +
      8 \theta_{22}^2)
  \bigg].
\end{split}
\end{equation}
We introduce the dimensionless quantity $\hat{\lambda} = (\lambda - u_1) /
\sqrt{\theta}$, and then the equation $\mathrm{det}(\lambda {\bf I} -
\hat{\bf M}) = 0$ becomes
\begin{equation} \label{eq:ev_1D}
\hat{\lambda} \left[ \hat{\lambda}^4 -
  \frac{2}{45} \frac{101 \theta_{11} + 16 \theta_{22}}{\theta} \hat{\lambda}^2 -
  \frac{96}{25} \frac{q_1}{\rho \theta^{3/2}} \hat{\lambda} +
  \frac{1}{15} \frac{53\theta_{11}^2 - 16 \theta_{11} \theta_{22} +
    8 \theta_{22}^2}{\theta^2}
\right] = 0.
\end{equation}
Consider the special case $\theta_{11} = \theta_{22} = \theta$ and $q_1
= 0$, which implies the fluid is in its local equilibrium, all solutions
of the above equation are
\begin{equation} \label{eq:ev_1D_eq}
\hat{\lambda}_{1,5} = \pm\sqrt{\frac{13 + \sqrt{94}}{5}}, \quad
\hat{\lambda}_{2,4} = \pm\sqrt{\frac{13 - \sqrt{94}}{5}}, \quad
\hat{\lambda}_3 = 0.
\end{equation}
Therefore, in this case, $\hat{\bf M}(\hat{\bw})$ has no multiple
eigenvalues, thus is real diagonalizable. If $(\theta_{11} -
\theta_{22}) / \theta$ and $q_1 / (\rho \theta^{3/2})$ are small
enough, the roots of \eqref{eq:ev_1D} are small perturbations of
\eqref{eq:ev_1D_eq}, which are still real and separable. This shows
that there is a hyperbolicity region for 1D moment system around the
thermodynamic equilibrium, and the Maxwell distribution is an interior
point of the hyperbolicity region. A precise depiction of the
hyperbolicity region can be found in \cite{Muller}.


\subsection{Lack of hyperbolicity of 3D Grad's 13-moment system}
\label{sec:Grad13_3D}
To the best of our knowledge, there has not been any published
investigation on the hyperbolicity of the full 3D Grad's system. One
may take it for granted that the full 3D case is similar as the 1D
case and there exists a neighbourhood of the equilibrium such that the
system is hyperbolic. Unfortunately, this is not true.  In this
section, we are going to show that Maxwellian is on the boundary of
the hyperbolicity region. The analysis below contains some tedious
calculations, which are carried out by the computer algebra system
Mathematica \cite{Mathematica}.

In the 3D case, Grad's 13-moment equations can also be written
in the quasi-linear form as
\begin{equation} \label{eq:Grad}
\pd{\bw}{t} +
  {\bf M}_k(\bw) \pd{\bw}{x_k} = \bQ(\bw).
\end{equation}
Now $\bw$ is a vector with 13 entries:
\begin{displaymath}
\bw = (\rho, u_1, u_2, u_3, \theta_{11}, \theta_{22}, \theta_{33},
  \theta_{12}, \theta_{13}, \theta_{23}, q_1, q_2, q_3)^T.
\end{displaymath}
The expressions of the matrices ${\bf M}_k$ and the operator $\bQ$ can
be obtained from \eqref{eq:G13}. Since Grad's moment system is
rotationally invariant, in order to check the hyperbolicity of \eqref
{eq:Grad}, we only need to check the diagonalizability of ${\bf M}_1$.
As a reference, the precise form of ${\bf M}_1(\bw)$ is given on page
\pageref{eq:M1}.

\begin{table}[p]
\centering
\noindent\rotatebox{-90}{\vbox{%
\small
\renewcommand\arraystretch{1.5}
\begin{displaymath} \label{eq:M1}
{\bf M}_1(\bw) = \left( \begin{array}{ccccccccccccc}
u_1 & \rho & 0 & 0 & 0 & 0 & 0 & 0 & 0 & 0 & 0 & 0 & 0 \\
\frac{\theta_{11}}{\rho} & u_1 &
  0 & 0 & 1 & 0 & 0 & 0 & 0 & 0 & 0 & 0 & 0 \\
\frac{\theta_{12}}{\rho} & 0 & u_1 &
  0 & 0 & 0 & 0 & 0 & 1 & 0 & 0 & 0 & 0 \\
\frac{\theta_{13}}{\rho} & 0 & 0 & u_1 &
  0 & 0 & 0 & 0 & 0 & 1 & 0 & 0 & 0 \\
0 & 2\theta_{11} & 0 & 0 & u_1 &
  0 & 0 & 0 & 0 & 0 & \frac{6}{5\rho} & 0 & 0 \\
0 & 0 & 2\theta_{12} & 0 & 0 & u_1 &
  0 & 0 & 0 & 0 & \frac{2}{5\rho} & 0 & 0 \\
0 & 0 & 0 & 2\theta_{13} & 0 & 0 & u_1 &
  0 & 0 & 0 & \frac{2}{5\rho} & 0 & 0 \\
0 & \theta_{12} & \theta_{11} & 0 &
  0 & 0 & 0 & u_1 & 0 & 0 & 0 & \frac{2}{5\rho} & 0 \\
0 & \theta_{13} & 0 & \theta_{11} &
  0 & 0 & 0 & 0 & u_1 & 0 & 0 & 0 & \frac{2}{5\rho} \\
0 & 0 & \theta_{13} & \theta_{12} &
  0 & 0 & 0 & 0 & 0 & u_1 & 0 & 0 & 0 \\
-(\theta - \theta_{11})^2 - (\theta_{12}^2 + \theta_{13}^2) &
  \frac{16q_1}{5} & \frac{2q_2}{5} & \frac{2q_3}{5} &
  \frac{\rho(\theta_{11} + 8\theta)}{6} &
  \frac{\rho(7\theta_{11} - 4\theta)}{6} &
  \frac{\rho(7\theta_{11} - 4\theta)}{6} &
  -\rho \theta_{12} & -\rho \theta_{13} & 0 & u_1 & 0 & 0 \\
\theta_{12}\theta_{33} - \theta_{13}\theta_{23} - \theta\theta_{12} &
  \frac{7q_2}{5} & \frac{7q_1}{5} & 0 &
  \frac{\rho \theta_{12}}{6} & 
  \frac{7\rho \theta_{12}}{6} & 
  \frac{7\rho \theta_{12}}{6} & 
  \rho(2\theta - \theta_{22}) & -\rho \theta_{23} & 0 & 0 & u_1 & 0 \\
\theta_{13}\theta_{22} - \theta_{12}\theta_{23} - \theta\theta_{13} &
  \frac{7q_3}{5} & 0 & \frac{7q_1}{5} &
  \frac{\rho \theta_{13}}{6} & 
  \frac{7\rho \theta_{13}}{6} & 
  \frac{7\rho \theta_{13}}{6} & 
  -\rho \theta_{23} & \rho(2\theta - \theta_{33}) & 0 & 0 & 0 & u_1
\end{array} \right)
\end{displaymath}}}
\end{table}

When $\bw$ represents the equilibrium state, which means
\begin{equation} \label{eq:Maxwellian}
\theta_{12} = \theta_{13} = \theta_{23} = q_1 = q_2 = q_3 = 0,
  \quad \theta_{11} = \theta_{22} = \theta_{33} = \theta.
\end{equation}
The characteristic polynomial of ${\bf M}_1(\bw)$ is
\begin{displaymath}
\mathrm{det}(\lambda {\bf I} - {\bf M}_1) = \frac{1}{125}
  (\lambda - u_1)^5 [5(\lambda - u_1)^2 - 7\theta]^2
  [5(\lambda - u_1)^4 - 26 \theta (\lambda - u_1)^2 + 15 \theta^2].
\end{displaymath}
All roots of the above polynomial are
\begin{displaymath}
u_1, \quad u_1 \pm \sqrt{\frac{7}{5} \theta}, \quad
u_1 \pm \sqrt{\frac{13 + \sqrt{94}}{5} \theta}, \quad
u_1 \pm \sqrt{\frac{13 - \sqrt{94}}{5} \theta}.
\end{displaymath}
Thus the eigenvalues of ${\bf M}_1$ are all real. In order to check
its diagonalizability, let
\begin{displaymath}
q(\lambda) =
  \frac{1}{25} (\lambda-u_1) [5(\lambda-u_1)^2 - 7\theta] \cdot
  [5(\lambda - u_1)^4 - 26 \theta (\lambda - u_1)^2 + 15 \theta^2].
\end{displaymath}
Direct verification shows $q({\bf M}_1) = {\bf 0}$. According to
Lemma \ref{thm:diag}, ${\bf M}_1(\bw)$ is real diagonalizable at the
equilibrium state.

In order to show that the equilibrium is on the boundary of the
hyperbolicity region, we consider the following case:
\begin{equation} \label{eq:Gaussian}
\theta_{13} = \theta_{23} = q_1 = q_2 = q_3 = 0, \quad
  \theta_{11} = \theta_{22} = \theta_{33} = \theta.
\end{equation}
When $f$ is the following Gaussian distribution:
\begin{equation} \label{eq:Gaussian_f}
f = \frac{\rho}{\sqrt{\mathrm{det}(2\pi\bTheta)}} \exp \left(
  -\frac{1}{2} \bC^T \bTheta^{-1} \bC
\right), \quad
\bTheta = \left( \begin{array}{ccc}
  \theta & \theta_{12} & 0 \\
  \theta_{12} & \theta & 0 \\
  0 & 0 & \theta
\end{array} \right),
\end{equation}
the relation \eqref{eq:Gaussian} is satisfied. When $|\theta_{12}| <
\theta$, the matrix $\bTheta$ is positive definite, and thus the
distribution function \eqref{eq:Gaussian_f} can be a physical
configuration. Substituting \eqref{eq:Gaussian} into \eqref{eq:M1}
and calculating the characteristic polynomial of ${\bf M}_1(\bw)$, one
has
\begin{displaymath}
\begin{split}
\mathrm{det}(\lambda {\bf I} - {\bf M}_1) &= \frac{1}{125}
  (\lambda - u_1)^3 [5(\lambda - u_1)^2 - 7\theta] \cdot
  r \left( \frac{(\lambda - u_1)^2}{\theta} \right),
\end{split}
\end{displaymath}
where
\begin{displaymath}
r(x) = 25x^4 - 165x^3 +
  \left( 257 + 48 \frac{\theta_{12}^2}{\theta^2} \right) x^2 +
  \left( 8 \frac{\theta_{12}^2}{\theta^2} - 105 \right) x -
  28 \frac{\theta_{12}^2}{\theta^2}.
\end{displaymath}
Let
\begin{displaymath}
\begin{split}
q(\lambda) &= 
  (\lambda - u_1) [5(\lambda - u_1)^2 - 7\theta] \cdot 
  r \left( \frac{(\lambda - u_1)^2}{\theta} \right).
\end{split}
\end{displaymath}
Obviously $q(\lambda)$ and $\mathrm{det}(\lambda {\bf I} - {\bf M}_1)$
share the same roots. Direct calculation of $q({\bf M}_1)$ gives us that
\begin{displaymath}
q({\bf M}_1) = \frac{56 \theta^2 \theta_{12}^3}{\rho}
  (\rho \theta {\bf E}_{10,4} - {\bf E}_{10,13}),
\end{displaymath}
where ${\bf E}_{i,j} = \be_i \be_j^T$, and $\be_j$ is the unit vector
with the $j$-th entry being $1$. According to Corollary \ref
{prop:non_diag}, if $\theta_{12} \neq 0$, then ${\bf M}_1(\bw)$ is not
diagonalizable. Actually, one may find that $r(x)$ have at least one
negative root since $r(-\infty) > 0$ and $r(0) < 0$, and therefore
${\bf M}_1(\bw)$ has eigenvalues with nonzero imaginary parts, which
also violates the hyperbolic condition.

The above analysis shows that when \eqref{eq:Gaussian} and
$\theta_{12} \neq 0$ holds, the hyperbolicity of \eqref{eq:Grad}
breaks down, no matter how small the value of $\theta_{12}$ is. It
turns out that there does not exist a neighbourhood of the equilibrium
such that all the states in this neighbourhood lead to the
hyperbolicity of Grad's 13-moment system. Without the hyperbolicity in
a neighbourhood of the equilibrium, the wellposedness of the Grad's
13-moment system is not guranteed even if the phase density is
extremely close to the equilibrium. This severe drawback may be the
possible reason why there are hardly any positive evidences for the
Grad's 13-moment system in the last decades.

\clearpage


\section{Modified 13-moment System} \label{sec:new_13m}
The results in Section \ref{sec:Grad13_3D} reveal a crucial issue of
Grad's original system. In order to establish the local hyperbolicity
around the equilibrium state, we derive a modified 13-moment system in
this section, which is hyperbolic for any states close enough to the
equilibrium. The proofs will be given in detail, and the size of the
hyperbolicity region will be discussed.

\subsection{Derivation of the modified system}
The modified 13-moment system is based on the following assumption of
the phase density:
\begin{equation} \label{eq:tilde_f_13}
\tilde{f}|_{13} = \left[
  1 + \frac{2}{5\rho} \bs^T \bTheta^{-1} \bC \left(
    \frac{1}{2} \bC^T \bTheta^{-1} \bC - \frac{5}{2}
  \right)
\right] f_G,
\end{equation}
where $\bs = (s_1, s_2, s_3)^T$, and $f_G$ is a Gaussian distribution:
\begin{equation} \label{eq:f_G}
f_G = \frac{\rho}{m\sqrt{\mathrm{det}(2\pi\bTheta)}} \exp \left(
  -\frac{1}{2} \bC^T \bTheta^{-1} \bC
\right).
\end{equation}
Comparing with $f_M$, the function $f_G$ incorporates the whole
temperature tensor into the exponent, and thus it can be expected that
such an approximation includes more nonlinearity than
\eqref{eq:dis_13m}, and is more suitable for describing anisotropic
density functions. In order to meet the requirement of orthogonality,
the vector $\bs$ should be related to the density function by
\begin{displaymath}
\bs = \frac{1}{2} \langle C_G^2 \bC f \rangle,
\end{displaymath}
where $C_G^2 = \bC^T \bTheta^{-1} \bC$. For the postulate
\eqref{eq:tilde_f_13}, the relation between $\bs$ and the heat flux
$\bq$ is
\begin{displaymath}
q_j = \frac{3}{5} s_{(i} \theta_{ij)}.
\end{displaymath}
Similar as the derivation of
Grad's 13-moment system, the new moment system can be written as
\begin{displaymath}
\left \langle \tilde{\bphi}
  \pd{\tilde{f}|_{13}}{t} \right \rangle +
\left \langle \tilde{\bphi}
  (\bxi \cdot \nabla_{\bx} \tilde{f}|_{13}) \right \rangle
= \left \langle
  \tilde{\bphi} Q(\tilde{f}|_{13},\tilde{f}|_{13}) \right\rangle,
\end{displaymath}
where 
\begin{displaymath}
\begin{split}
\tilde{\bphi} ={} & (1, \\
               & ~C_1, C_2, C_3, \\
               & ~C_1^2, C_2^2, C_3^2, C_1 C_2, C_1 C_3, C_2 C_3, \\
               & ~C_G^2 C_1, C_G^2 C_2, C_G^2 C_3)^T.
\end{split}
\end{displaymath}
We reformulate the resulting system in explicit form as
\begin{subequations} \label{eq:new_13m}
\begin{align}
& \frac{\mathrm{d} \rho}{\mathrm{d} t} +
  \rho \pd{u_k}{x_k} = 0, \\
& \frac{\mathrm{d} u_i}{\mathrm{d} t} +
  \frac{\theta_{ik}}{\rho} \pd{\rho}{x_k} +
  \pd{\theta_{ik}}{x_k} = 0, \quad i = 1,2,3, \\
& \frac{\mathrm{d} \theta_{ij}}{\mathrm{d} t} +
  2 \theta_{k(i} \pd{u_{j)}}{x_k} +
  \frac{6}{5\rho} \left(
    s_{(i} \pd{\theta_{jk)}}{x_k} +
    \theta_{(ij} \pd{s_{k)}}{x_k}
  \right) =
    -\frac{\rho\theta}{\mu} (\theta_{ij} - \delta_{ij} \theta),
  \quad i,j = 1,2,3, \\
\label{eq:s_j}
\begin{split}
& \frac{\mathrm{d} s_j}{\mathrm{d} t} -
  \frac{6}{5} \theta^{ik} \theta_{l(i} s_j
    \pd{u_{k)}}{x_l} -
  \frac{18}{25\rho} \theta^{ik} \left(
    s_{(i} s_j \pd{\theta_{kl)}}{x_l} +
    s_{(i} \theta_{jk} \pd{s_{l)}}{x_l}
  \right) \\
& \qquad + \frac{1}{2} \left(
    \rho \theta^{ik} \theta_{jl}
      \pd{\theta_{ik}}{x_l} +
    2\rho \pd{\theta_{jl}}{x_l}
  \right) + \frac{2}{5} \left(
    7 s_{(j} \pd{u_{l)}}{x_l} +
    \theta^{ik} \theta_{jl} s_{(i}
      \pd{u_{k)}}{x_l}
  \right) = \tilde{Q}_j, \quad j = 1,2,3.
\end{split}
\end{align}
\end{subequations}
In equation \eqref{eq:s_j}, $\theta^{ij}$ stands for the $(i,j)$
entry of matrix $\bTheta^{-1}$, and
\begin{displaymath}
\tilde{Q}_j = -\frac{\rho \theta}{\mu} \left(
  \frac{71}{30} s_j - \frac{9}{10} \theta \theta^{(ii} s_{j)}
    - \frac{1}{15} \theta^{ii} \theta_{jk} s_k
\right).
\end{displaymath}
The expressions of $\tilde{Q}_j$ are obtained by using the following
properties of Maxwell molecules:
\begin{displaymath}
\begin{gathered}
\langle C_i C_j Q(f,f) \rangle = -\frac{\rho \theta}{\mu}
  \left \langle
    \left( C_i C_j - \frac{1}{3} C^2 \delta_{ij} \right) f
  \right \rangle, \quad
\langle C^2 C_j Q(f,f) \rangle =
  -\frac{2}{3} \frac{\rho \theta}{\mu} \langle C^2 C_j f \rangle, \\
\left \langle \left(
  C_i C_j C_k - \frac{3}{5} C^2 C_{(i} \delta_{jk)}
\right) Q(f,f) \right \rangle =
-\frac{3}{2} \frac{\rho \theta}{\mu} \left \langle \left(
  C_i C_j C_k - \frac{3}{5} C^2 C_{(i} \delta_{jk)}
\right) f \right \rangle.
\end{gathered}
\end{displaymath}

The system \eqref{eq:new_13m} can also be written in a quasi-linear
form:
\begin{equation} \label{eq:new_13m_mat}
\pd{\tilde{\bw}}{t} +
  \tilde{\bf M}_k(\tilde{\bw})
    \pd{\tilde{\bw}}{x_k} =
  \tilde{\bQ}(\tilde{\bw}),
\end{equation}
and we choose
\begin{displaymath}
\tilde{\bw} = (\rho, u_1, u_2, u_3, \theta_{11}, \theta_{22},
  \theta_{33}, \theta_{12}, \theta_{13}, \theta_{23}, s_1, s_2, s_3)^T.
\end{displaymath}
Since the linear space spanned by all the components of
$\tilde{\bphi}$ is rotationally invariant, the moment equations \eqref
{eq:new_13m} are also rotationally invariant. Therefore, below we
focus on the first coefficient matrix $\tilde{\bf M}_1(\tilde{\bw})$.

\subsection{Local hyperbolicity of the modified system}
Before establishing the local hypebolicity of \eqref{eq:new_13m}, we
provide a technical lemma first:

\begin{lemma} \label{lem:equiv}
For a given symmetric positive definite matrix $\bTheta =
(\theta_{ij})_{3\times 3}$, the inequality
\begin{equation} \label{eq:ineq}
\theta_{11}^{-1} s_1^2 \leqslant \bs^T \bTheta^{-1} \bs
\end{equation}
holds for any vector $\bs = (s_1, s_2, s_3)^T \in \mathbb{R}^3$. The
equality holds if and only if there exists a constant $k$ such that
\begin{equation} \label{eq:equivalence}
s_1 = k \theta_{11}, \quad s_2 = k \theta_{12},
  \quad s_3 = k \theta_{13}.
\end{equation}
\end{lemma}
\begin{proof}
We first prove $\theta_{12} \theta^{12} + \theta_{13} \theta^{13}
\leqslant 0$. Let $\mathfrak{S} = \theta_{12} \theta^{12} +
\theta_{13} \theta^{13}$. Then we have
\begin{equation} \label{eq:linear}
\theta_{12} (\theta_{23} \theta_{13} - \theta_{12} \theta_{33}) +
  \theta_{13} (\theta_{12} \theta_{23} - \theta_{22} \theta_{13}) =
  \mathfrak{S} \mathrm{det}(\bTheta).
\end{equation}
This equation can be considered as a quadratic equation of
$\theta_{13}$, and its discriminant is
\begin{displaymath}
\Delta = (2\theta_{12} \theta_{23})^2 - \theta_{22}
  [4 \theta_{12}^2 \theta_{33} + \mathfrak{S} \mathrm{det}(\bTheta)]
= 4\theta_{12}^2 (\theta_{23}^2 - \theta_{22} \theta_{33}) -
  \mathfrak{S} \theta_{22} \mathrm{det}(\bTheta).
\end{displaymath}
Since $\bTheta$ is positive definite, the following inequalities hold:
\begin{displaymath}
\theta_{23}^2 - \theta_{22} \theta_{33} < 0, \quad
\theta_{22} > 0, \quad \mathrm{det}(\bTheta) > 0.
\end{displaymath}
Thus, in order that \eqref{eq:linear} is not less than zero,
$\mathfrak{S} \leqslant 0$ must hold. Moreover, if $\mathfrak{S} = 0$,
$\theta_{12}$ must also be zero, and then \eqref{eq:linear} becomes
$\theta_{22} \theta_{13}^2 = 0$, which means $\theta_{13} = 0$.
Obviously when $\theta_{12} = \theta_{13} = 0$, one has $\mathfrak{S}
= 0$. Therefore we finally conclude that $\mathfrak{S} \leqslant 0$,
and the equality holds if and only if $\theta_{12} = \theta_{13} = 0$.

Now let $\mathfrak{D} = \bs^T \bTheta^{-1} \bs - \theta_{11}^{-1}
s_1^2$, which can be written as
\begin{equation} \label{eq:cond}
(\theta_{11} \theta^{11} - 1) s_1^2 +
  2\theta_{11} (\theta^{12} s_2 + \theta^{13} s_3) s_1 +
  \theta_{11}(\theta^{22} s_2^2+2\theta^{23} s_2 s_3+\theta^{33} s_3^2)
= \theta_{11} \mathfrak{D}.
\end{equation}
We consider the following two cases:
\begin{itemize}
\item If $\theta_{11} \theta^{11} - 1 = 0$, since $\theta_{1k}
\theta^{1k} = 1$, one has $\theta_{12} \theta^{12} + \theta_{13}
\theta^{13} = 0$. In this case, $\theta_{12} = \theta_{13} = 0$, and
therefore
\begin{displaymath}
\theta^{12} = \frac{\theta_{13} \theta_{23} - \theta_{12} \theta_{33}}
  {\mathrm{det} (\bTheta)} = 0, \quad
\theta^{13} = \frac{\theta_{12} \theta_{23} - \theta_{13} \theta_{22}}
  {\mathrm{det} (\bTheta)} = 0.
\end{displaymath}
Thus \eqref{eq:cond} becomes $\theta^{22} s_2^2 + 2\theta^{23} s_2 s_3
+ \theta^{33} s_3^2 = \mathfrak{D}$, which is equivalent to $\tilde{\bs}^T
\bTheta^{-1} \tilde{\bs} = \mathfrak{D}$ for $\tilde{\bs} = (0, s_2, s_3)^T$.
Since $\bTheta$ is positive definite, $\bTheta^{-1}$ is also positive
definite. This shows that $\mathfrak{D} \geqslant 0$, and the equality
holds if and only if $s_2 = s_3 = 0$. In the case of $\mathfrak{D} =
s_2 = s_3 = \theta_{12} = \theta_{13} = 0$, the relation
\eqref{eq:equivalence} holds with $k = s_1 / \theta_{11}$.
\item If $\theta_{11} \theta^{11} - 1 \neq 0$, then $\theta_{11}
\theta^{11} - 1 = -(\theta_{12} \theta^{12} + \theta_{13} \theta^{13})
> 0$. In this case, \eqref{eq:cond} is a quadratic equation of $s_1$,
whose discriminant is
\begin{displaymath}
\begin{split}
\Delta &= [2\theta_{11} (\theta^{12} s_2 + \theta^{13} s_3)]^2 -
  4\theta_{11} (\theta_{11} \theta^{11} - 1)
    (\theta^{22} s_2^2 + 2\theta^{23} s_2 s_3 + \theta^{33} s_3^2 -
     \mathfrak{D}) \\
&= -4 \theta_{11} (\theta_{12} s_3 - \theta_{13} s_2)^2 /
  \mathrm{det} (\bTheta) +
  4 \theta_{11} \mathfrak{D} (\theta_{11} \theta^{11} - 1).
\end{split}
\end{displaymath}
In order that $s_1$ is real, $\mathfrak{D} \geqslant 0$ must hold. And
if $\mathfrak{D} = 0$, $\theta_{12} s_3 - \theta_{13} s_2$ must be
zero. Thus when $\mathfrak{D}$ is zero, there exist a constant $k$
such that
\begin{equation} \label{eq:s_23}
s_2 = k \theta_{12}, \qquad s_3 = k \theta_{13}.
\end{equation}
Substitute \eqref{eq:s_23} and $\mathfrak{D} = 0$ into
\eqref{eq:cond}, it can be solved that $s_1 = k \theta_{11}$.
\end{itemize}
In both cases, \eqref{eq:ineq} holds, and it has been demonstrated
that if $\theta_{11}^{-1} s_1^2 = \bs^T \bTheta^{-1} \bs$, then
\eqref{eq:equivalence} holds. It only remains to prove that
$\theta_{11}^{-1} s_1^2 = \bs^T \bTheta^{-1} \bs$ is a conclusion of
\eqref{eq:equivalence}.

If \eqref{eq:equivalence} holds, then
\begin{displaymath}
(k, 0, 0) \bTheta = k(\theta_{11}, \theta_{12}, \theta_{13}) = \bs^T.
\end{displaymath}
Therefore $\bs^T \bTheta^{-1} = (k, 0, 0)$, and then
\begin{equation}
\bs^T \bTheta^{-1} \bs = (k, 0, 0) \bs
  = k^2 \theta_{11} = \theta_{11}^{-1} s_1^2.
\end{equation}
This completes the proof of the lemma.
\end{proof}

Now we claim that the modified 13-moment system \eqref{eq:new_13m} is
locally hyperbolic around the equilibrium. Precisely, we have the
following major theorem of this section:
\begin{theorem} \label{prop:1D_hyp}
There exists a positive constant $\delta > 0$, such that if $\rho^{-2}
\bs^T \bTheta^{-1} \bs < \delta$, $\tilde{\bf M}_1(\tilde{\bw})$ is
real diagonalizable.
\end{theorem}
\begin{proof}
Let
\begin{displaymath}
\eta_1 := \rho^{-2} \theta_{11}^{-1} s_1^2, \quad
\eta_2 := \rho^{-2} \bs^T \bTheta^{-1} \bs, \quad
\zeta = \frac{\lambda - u_1}{\sqrt{\theta_{11}}}.
\end{displaymath}
According to Lemma \ref{lem:equiv}, we have $\eta_1 \leqslant \eta_2
< \delta$. By direct calculation, the characteristic polynomial of
$\tilde{\bf M}_1(\bw)$ is
\begin{equation} \label{eq:cp}
p(\lambda) := \mathrm{det}(\lambda {\bf I} - \tilde{\bf M}_1) =
  \frac{(\sqrt{\theta_{11}})^{13}}{1953125} 
    [p_{11}(\zeta) + \eta_2 p_{12} (\zeta)]
    [p_{21}(\zeta) + \eta_2 p_{22} (\zeta)],
\end{equation}
where
\begin{equation} \label{eq:p}
\begin{aligned}
p_{11}(\zeta) &= 25 \zeta^2 (5\zeta^2 - 7)
  - 130 \sqrt{\eta_1} \zeta^3 + 4 \eta_1 (6\zeta^2 + 7), \qquad
p_{12}(\zeta) = 8 \zeta^2, \\
p_{21}(\zeta) &= 625 \zeta^3
  (25 \zeta^6 - 165 \zeta^4 + 257 \zeta^2 - 105) \\
& \quad -250 \eta_1^{1/2} \zeta^2
  (110 \zeta^6 - 311 \zeta^4 + 144 \zeta^2 - 105) \\
& \quad +100 \eta_1 \zeta
  (111 \zeta^6 + 447 \zeta^4 - 209 \zeta^2 + 105) \\
& \quad -40 \eta_1^{3/2}
  (18 \zeta^6 + 697 \zeta^4 + 282 \zeta^2 + 105) \\
& \quad + 96 \eta_1^2 \zeta (16 \zeta^2 + 63), \\
p_{22}(\zeta) &= 8 \zeta \big[
  25 \zeta^2 (23 \zeta^4 - 73 \zeta^2 + 3)
  -10 \sqrt{\eta_1} \zeta (27\zeta^4 + 67\zeta^2 - 18)
  + 12 \eta_1 (48 \zeta^2 - 7) \big].
\end{aligned}
\end{equation}
Below we divide the proof into three cases.

\subparagraph{First case: $\boldsymbol{s_1 = s_2 = s_3 = 0}$.}
In this case, $\eta_1 = \eta_2 = 0$, and
\begin{displaymath}
p(\lambda) = \frac{(\sqrt{\theta_{11}})^{13}}{125} \zeta^5
  (5\zeta^2 - 7)^2 (5\zeta^4 - 26 \zeta^2 + 15).
\end{displaymath}
Obviously the roots of $p(\lambda)$ are all real.  According to
Lemma \ref{thm:diag}, we only need to prove $\hat{p}(\tilde{\bf
M}_1) = {\bf 0}$ for
\begin{displaymath}
\hat{p}(\lambda) = (\sqrt{\theta_{11}})^{13} \zeta
  (5\zeta^2 - 7)(5\zeta^4 - 26 \zeta^2 + 15).
\end{displaymath}
This can be verified by direct calculation.

\subparagraph{Second case: $\boldsymbol{s_1 = 0}$ and
$\boldsymbol{s_2^2 + s_3^2 > 0}$.}
In this case, $\eta_1 = 0$, while the SPD property of $\bTheta$ gives
$\eta_2 > 0$. The characteristic polynomial $p(\lambda)$ can be
simplified as
\begin{displaymath}
p(\lambda) = \frac{\zeta^5}{78125} p_1(\zeta) p_2(\zeta),
\end{displaymath}
where
\begin{displaymath}
\begin{split}
p_1(\zeta) &= 25 (5\zeta^2 - 7) + 8 \eta_2, \\
p_2(\zeta) &= 25 (5\zeta^2 - 7) (5\zeta^4 - 26 \zeta^2 + 15)
  + 8 \eta_2 (23 \zeta^4 - 73 \zeta^2 + 3).
\end{split}
\end{displaymath}
When $\eta_2$ equals zero, all the roots of $p_1$ and $p_2$ are
single and nonzero. Thus, when $\eta_2 < \delta$ for $\delta$ small
enough, the roots of $p_1$ and $p_2$ are also single and nonzero.
Furthermore, we claim that $p_1(\zeta)$ and $p_2(\zeta)$ have no
common roots when $\delta$ is small enough. This can be proven
following these steps:
\begin{enumerate}
\item Let $\tilde{p}_1(z) = p_1(\sqrt{z})$, $\tilde{p}_2(z) =
p_2(\sqrt{z})$. Obviously, when $\delta$ is small enough,
$\tilde{p}_1$ and $\tilde{p}_2$ are polynomials with all their roots
positive. If $\tilde{p}_1$ and $\tilde{p}_2$ have no common roots,
then $p_1$ and $p_2$ have no common roots.
\item The polynomial $\tilde{p}_1(z)$ is a linear function, and its
only root is $(175 - 8\eta_2) / 125$. The value of $\tilde{p}_2$ at
this point is
\begin{equation} \label{eq:remainder}
\tilde{p}_2 \left( \frac{175 - 8 \eta_2}{125} \right) =
  \frac{8\eta_2 (1152 \eta_2^2 - 3400 \eta_2 - 664375)}{15625}.
\end{equation}
\item Since $0 < \eta_2 < \delta$, the value of \eqref{eq:remainder}
is negative if $\delta < (425 + 125\sqrt{3073}) / 288$, which means
$\tilde{p}_1$ and $\tilde{p}_2$ have no common roots.
\end{enumerate}
The above analysis shows when $\delta$ is small, $p_1(\zeta)
p_2(\zeta)$ has no multiple roots, and $\zeta = 0$ is not a root of
$p_1(\zeta) p_2(\zeta)$. Thus, the polynomial
\begin{equation} \label{eq:q}
q(\zeta) := \zeta p_1(\zeta) p_2(\zeta)
\end{equation}
has no multiple roots if $\delta$ is small. Finally, it is verified by
computer algebra system
\begin{equation} \label{eq:q_M}
q\left( \frac{\tilde{\bf M}_1-u_1 {\bf I}}{\sqrt{\theta_{11}}} \right)
  = {\bf 0}, \qquad \text{if} \quad s_1 = 0.
\end{equation}
Hence, according the Lemma \ref{thm:diag}, the matrix $\tilde{\bf
M}_1$ is diagonalizable.

\subparagraph{Third case: $\boldsymbol{s_1 \neq 0}$.}
In this case, $\eta_1 > 0$ and $\eta_2 > 0$. We first prove when
$\delta$ is small, both $p_{11}(\zeta) + \eta_2 p_{12}(\zeta)$ and
$p_{21}(\zeta) + \eta_2 p_{22}(\zeta)$ have no multiple or imaginary
roots. When $\eta_1 = \eta_2 = 0$,
\begin{subequations}
\begin{align}
\label{eq:p1}
p_{11}(\zeta) + \eta_2 p_{12}(\zeta) &= 25\zeta^2 (5\zeta^2 - 7), \\
\label{eq:p2}
p_{21}(\zeta) + \eta_2 p_{22}(\zeta) &=
  625 \zeta^3 (5\zeta^2 - 7) (5\zeta^4 - 26\zeta^2 + 15).
\end{align}
\end{subequations}
Both polynomials have only real roots, and both of them have only one
multiple roots --- $\zeta = 0$. Thus, for small $\delta$, if $\eta_1$
and $\eta_2$ are nonzero, then the multiple or imaginary roots must be
around $\zeta = 0$ if they exist. For $p_{11}(\zeta) + \eta_2
p_{12}(\zeta)$, when $\delta$ is small, one can obtain its values at
some particular points around $\zeta = 0$:
\begin{itemize}
\item $\zeta = -\sqrt{\eta_1}$: $\quad p_{11}(-\sqrt{\eta_1}) + \eta_2
  p_{12}(-\sqrt{\eta_1}) = 279 \eta_1^2 - 147 \eta_1 + 8 \eta_1 \eta_2
  < 0$,
\item $\zeta = 0$: $\quad p_{11}(0) + \eta_2 p_{12}(0) = 28 \eta_1 > 0$,
\item $\zeta = \sqrt{\eta_1}$: $\quad p_{11}(\sqrt{\eta_1}) + \eta_2
  p_{12}(\sqrt{\eta_1}) = 19 \eta_1^2 - 147 \eta_1 + 8 \eta_1 \eta_2 <
  0$.
\end{itemize}
This tells us that there are two distinct real roots of
$p_{11}(\zeta) + \eta_2 p_{12}(\zeta)$ around $\zeta = 0$. Noting that
$\zeta = 0$ is a root of multiplicity $2$ of \eqref{eq:p1}, we
conclude that in the case of $0 < \eta_1 < \delta$ and $0 < \eta_2 <
\delta$, $p_{11}(\zeta) + \eta_2 p_{12}(\zeta)$ has no multiple or
imaginary roots. Similarly, for $p_{11}(\zeta) + \eta_2
p_{12}(\zeta)$, when $\delta$ is small, one has
\begin{itemize}
\item $\zeta = -\sqrt{\eta_1}$:
\begin{displaymath}
\begin{split}
p_{21}(-\sqrt{\eta_1}) + \eta_2 p_{22}(-\sqrt{\eta_1}) &=
-\eta_1^{3/2}
  [54945 \eta_1^3 - (106759 - 6760\eta_2) \eta_1^2 \\
& \quad +
  (193053 - 4632\eta_2) \eta_1 - (77175 + 1512\eta_2)] > 0,
\end{split}
\end{displaymath}
\item $\zeta = 0$:
\begin{displaymath}
p_{21}(0) + \eta_2 p_{22}(0) = -4200 \eta_1^{3/2} < 0,
\end{displaymath}
\item $\zeta = \frac{2}{5} \eta_1^{1/2} - \frac{4}{375} \eta_1^{3/2}$:
\begin{displaymath}
\begin{split}
& p_{21} \left( \frac{2}{5} \eta_1^{1/2} -
  \frac{4}{375} \eta_1^{3/2} \right)
+ \eta_2 p_{22} \left( \frac{2}{5} \eta_1^{1/2} -
  \frac{4}{375} \eta_1^{3/2} \right) \\
={} & \frac{64 \eta_1^{7/2}}{9385585784912109375}
  \big[ -4096\eta_1^{10} + 706560\eta_1^9 - 30182400\eta_1^8 \\
& \quad - 57600 (29025 + 184 \eta_2) \eta_1^7 +
  4320000 (48475 + 536 \eta_2) \eta_1^6 \\
& \quad - 486000000(18575 + 428 \eta_2) \eta_1^5 +
  506250000(539275 + 19784 \eta_2) \eta_1^4 \\
& \quad - 18984375000(412025 + 16176 \eta_2)\eta_1^3 +
  711914062500(209175 + 10576 \eta_2) \eta_1^2 \\
& \quad - 40045166015625(26225 + 3704 \eta_2) \eta_1 +
  3003387451171875(175 + 484 \eta_2) \big] > 0,
\end{split}
\end{displaymath}
\item $\zeta = \sqrt{\eta_1}$:
\begin{displaymath}
\begin{split}
p_{21}(\sqrt{\eta_1}) + \eta_2 p_{22}(\sqrt{\eta_1}) &=
-\eta_1^{3/2}
  [1495 \eta_1^3 + (7019 - 2440\eta_2) \eta_1^2 \\
& \quad -
  (98493 - 15352\eta_2) \eta_1 + (33075 + 1368\eta_2)] < 0.
\end{split}
\end{displaymath}
\end{itemize}
This reveals that there are three distinct real roots of
$p_{21}(\zeta) + \eta_2 p_{22}(\zeta)$ around $\zeta = 0$. Until now,
the statement at the beginning of this paragraph has been proven.

The subsequent proof is divided into two parts:
\begin{enumerate}
\item If $\eta_1 = \eta_2$, then $p_{11}(\zeta) + \eta_2
p_{12}(\zeta)$ is a factor of $p_{21}(\zeta) + \eta_2 p_{22}(\zeta)$,
and we actually have
\begin{align}
p_{11}(\zeta) + \eta_2 p_{12}(\zeta) &=
  (25\zeta^3 - 16 \sqrt{\eta_1} \zeta^2 - 35 \zeta - 14 \sqrt{\eta_1})
  (5 \zeta - 2\sqrt{\eta_1}), \\
\label{eq:p_2}
\begin{split}
p_{21}(\zeta) + \eta_2 p_{22}(\zeta) &=
  [25 \zeta (\zeta^4 - 26 \zeta^2 + 15) +
   30 \sqrt{\eta_1} (3\zeta^4 + 6\zeta^2 + 5) - 192 \eta_1 \zeta]
  \times {} \\
& \qquad (25\zeta^3 - 16 \sqrt{\eta_1} \zeta^2
  - 35 \zeta - 14 \sqrt{\eta_1}) (5 \zeta - 2\sqrt{\eta_1}).
\end{split}
\end{align}
Thus we need to verify
\begin{equation} \label{eq:p2_M1}
p_{21} \left(
  \frac{\tilde{\bf M}_1 - u_1 {\bf I}}{\sqrt{\theta_{11}}}
\right) + \eta_2 p_{22} \left(
  \frac{\tilde{\bf M}_1 - u_1 {\bf I}}{\sqrt{\theta_{11}}}
\right) = {\bf 0}.
\end{equation}
According to Lemma \ref{lem:equiv}, the condition $\eta_1 = \eta_2$ is
equivalent to \eqref{eq:equivalence}. Substitute
\eqref{eq:equivalence} into the expression of $\tilde{\bf M}_1$, and
\eqref{eq:p2_M1} then can be directly verified.
\item If $\eta_1 \neq \eta_2$, the resultant of $p_{11} + \eta_2
p_{12}$ and $p_{21} + \eta_2 p_{22}$ is calculated as
\begin{equation} \label{eq:resultant}
\mathrm{res}(p_{11} + \eta_2 p_{12}, p_{21} + \eta_2 p_{22}) =
  -1003520000000000 \eta_1^3 (\eta_1-\eta_2)^5 r(\eta_1, \eta_2),
\end{equation}
where $r(\eta_1, \eta_2)$ is
\begin{displaymath}
\begin{split}
r(\eta_1,\eta_2) = {} & 6519382474752 \eta_1^5 +
  7205633261568 \eta_2 \eta_1^4 -
  1047028571136000 \eta_1^4 + {}\\
& 2877437509632 \eta_2^2 \eta_1^3 +
  71846341632000 \eta_2 \eta_1^3 +
  6117273120960000 \eta_1^3 + {}\\
& 488268103680 \eta_2^3 \eta_1^2 +
  14075065958400 \eta_2^2 \eta_1^2 -
  32261927040000 \eta_2 \eta_1^2 - {}\\
& 12991498038500000 \eta_1^2 +
  31436439552 \eta_2^4 \eta_1 +
  74226585600 \eta _2^3 \eta_1 - {}\\
& 29723348160000 \eta_2^2 \eta_1 -
  84800409000000 \eta_2 \eta_1 +
  12363509395312500 \eta_1 + {}\\
& 668860416 \eta_2^5 - 13801881600 \eta_2^4 -
  707492160000 \eta_2^3 + {}\\
& 13556709000000 \eta_2^2 +
  188918353125000 \eta_2-3277351494140625.
\end{split}
\end{displaymath}
Evidently when $\delta$ is small, $r(\eta_1, \eta_2) < 0$. Noting that
$\eta_1 > 0$ and $\eta_1 \neq \eta_2$, we conclude
\eqref{eq:resultant} is nonzero. According to Lemma \ref{lem:coprime},
$p_{11}(\zeta) + \eta_2 p_{12}(\zeta)$ and $p_{21}(\zeta) + \eta_2
p_{22}(\zeta)$ have no common roots. Thus, the characteristic
polynomial $p(\lambda)$ has no multiple roots, which gives us the
diagonalizability of $\tilde{\bf M}_1$.
\end{enumerate}

\subparagraph{Final conclusion.}
For all the three cases listed above, it has been proven that when
$\delta$ is small, all the eigenvalues of $\tilde{\bf M}_1$ are real,
and the matrix $\tilde{\bf M}_1$ is diagonalizable. Thus the proof of
Theorem \ref{prop:1D_hyp} is completed.
\end{proof}

\begin{theorem} \label{thm:hyp_3D}
There exists a positive constant $\delta > 0$, such that if $\rho^{-2}
\bs^T \bTheta^{-1} \bs < \delta$, the moment system \eqref{eq:new_13m}
is hyperbolic.
\end{theorem}

\begin{proof}
The hyperbolicity of the moment system \eqref{eq:new_13m} is
equivalent to the diagonalizability of the matrix $n_k \tilde{\bf M}_k
(\tilde{\bw})$ for all unit vectors $\bn = (n_1, n_2, n_3)^T \in
\mathbb{R}^3$. The rotational invariance of \eqref{eq:new_13m} implies
that for any unit vector $\bn$, there exists a constant square matrix
$\bf R$ such that
\begin{displaymath}
n_k \tilde{\bf M}_k(\tilde{\bw}) =
  {\bf R}^{-1} \tilde{\bf M}_1({\bf R} \tilde{\bw}) {\bf R}.
\end{displaymath}
Actually, $\bf R$ can be constructed as follows:
\begin{enumerate}
\item Construct an orthogonal matrix ${\bf G} = (g_{ij})_{3\times 3}$
such that the first row of $\bf G$ is $(n_1, n_2, n_3)$.
\item Define the ``rotated moments'' $\tilde{\bw}'$ as
\begin{displaymath}
\tilde{\bw}' = (\rho', u_1', u_2', u_3',
  \theta_{11}', \theta_{22}', \theta_{33}',
  \theta_{12}', \theta_{13}', \theta_{23}',
  s_1', s_2', s_3')^T,
\end{displaymath}
where
\begin{equation} \label{eq:rotation}
\rho' = \rho, \quad u_i' = g_{ij} u_j, \quad
\theta_{ij}' = g_{ik} g_{jl} \theta_{kl}, \quad
s_i' = g_{ij} s_j.
\end{equation}
\item The matrix $\bf R$ is the unique matrix such that $\tilde{\bw}'
= {\bf R} \tilde{\bw}$ for all $\tilde{\bw}$.
\end{enumerate}
According to Theorem \ref{prop:1D_hyp}, there exists a constant
positive number $\delta$ such that the matrix $n_k \tilde{\bf
M}_k(\tilde{\bw})$ is diagonalizable if $\rho'^{-2} \bs'^T
(\bTheta')^{-1} \bs' < \delta$, where $\bTheta' =
(\theta_{ij}')_{3\times 3}$. Using \eqref{eq:rotation}, we have
\begin{displaymath}
\rho'^{-2} \bs'^T (\bTheta')^{-1} \bs'
  = \rho^{-2} ({\bf G}\bs)^T
    ({\bf G} {\bf\Theta} {\bf G}^T)^{-1} ({\bf G}\bs)
  = \rho^{-2} \bs^T {\bf \Theta}^{-1} \bs.
\end{displaymath}
Thus the theorem is proven.
\end{proof}

\subsection{Quantification of the hyperbolicity region}

The proof of Theorem \ref{thm:hyp_3D} reveals that the maximal value
of $\delta$ (denoted by $\delta_{\max}$ below) in Theorem
\ref{thm:hyp_3D} equals that in Theorem \ref{prop:1D_hyp}. Below we
give a rough estimation of $\delta_{\max}$. Let \[ \tilde{p} =
(p_{11}+\eta_2 p_{12})(p_{21}+\eta_2 p_{22}),\] where $p_{11}, p_{12},
p_{21}, p_{22}$ are defined in \eqref{eq:p}. We denote the domain on
which the polynomial $\tilde{p}$ has no imaginary roots to be
$\Sigma$, and thus
\[ 
\Sigma = \{ (\eta_1, \eta_2) \mid \mathfrak{I}(\eta_1, \eta_2) = 0 \},
\]
where
\[
\mathfrak{I}(\eta_1, \eta_2) :=
  \max \{ |\mathrm{Im}(z)| \mid \text{$z$ is the root of $\tilde{p}$}
\}, \qquad 0 \leqslant \eta_1 \leqslant \eta_2.
\]
Since $\mathfrak{I}$ is continuous, $\Sigma$ has to be a closed
region. We plot the domain $\Sigma$ as the green area in Figure
\ref{fig:Sigma}. The horizontal line $\eta_2 = \tilde{\delta}$ is
tangent to the red curve. We have $\delta_{\max} \leqslant
\tilde{\delta}$ and $\tilde{\delta} \approx 0.095$.

We denote the domain $\mathcal{S}$ to be the domain on which the
polynomial $\tilde{p}$ has multiple roots. According to Lemma
\ref{lem:multi_root} and Lemma \ref{lem:coprime}, we have that
\[
\mathcal{S} =
  \{ (\eta_1, \eta_2) \mid \mathfrak{R}(\eta_1, \eta_2) = 0 \},
\]
and
\[
\mathfrak{R}(\eta_1, \eta_2) := \mathrm{res}(\tilde{p},
\tilde{p}'), \quad \tilde{p}'(\zeta) = \odd{\zeta} \tilde{p}(\zeta),
\qquad 0 \leqslant \eta_1 \leqslant \eta_2.
\]

Due to the continuity of the roots of polynomials with respect to its
coefficients, we have $\partial \Sigma \subset \mathcal{S}$. Figure
\ref{fig:zls} shows part of $\mathcal{S}$. Comparing Figure
\ref{fig:Sigma} and Figure \ref{fig:zls}, we conclude that if $0 <
\eta_1 < \eta_2 < \tilde{\delta}$, which implies that $(\eta_1,
\eta_2)$ is an interior point of $\Sigma$ below the line $\eta_2 =
\tilde{\delta}$, then $\tilde{\bf M}_1(\tilde{\bw})$ is real
diagonalizable.
\begin{figure}[!ht]
\centering
\subfigure[The region $\Sigma$]{%
\label{fig:Sigma}
\begin{overpic}[width=.43\textwidth]{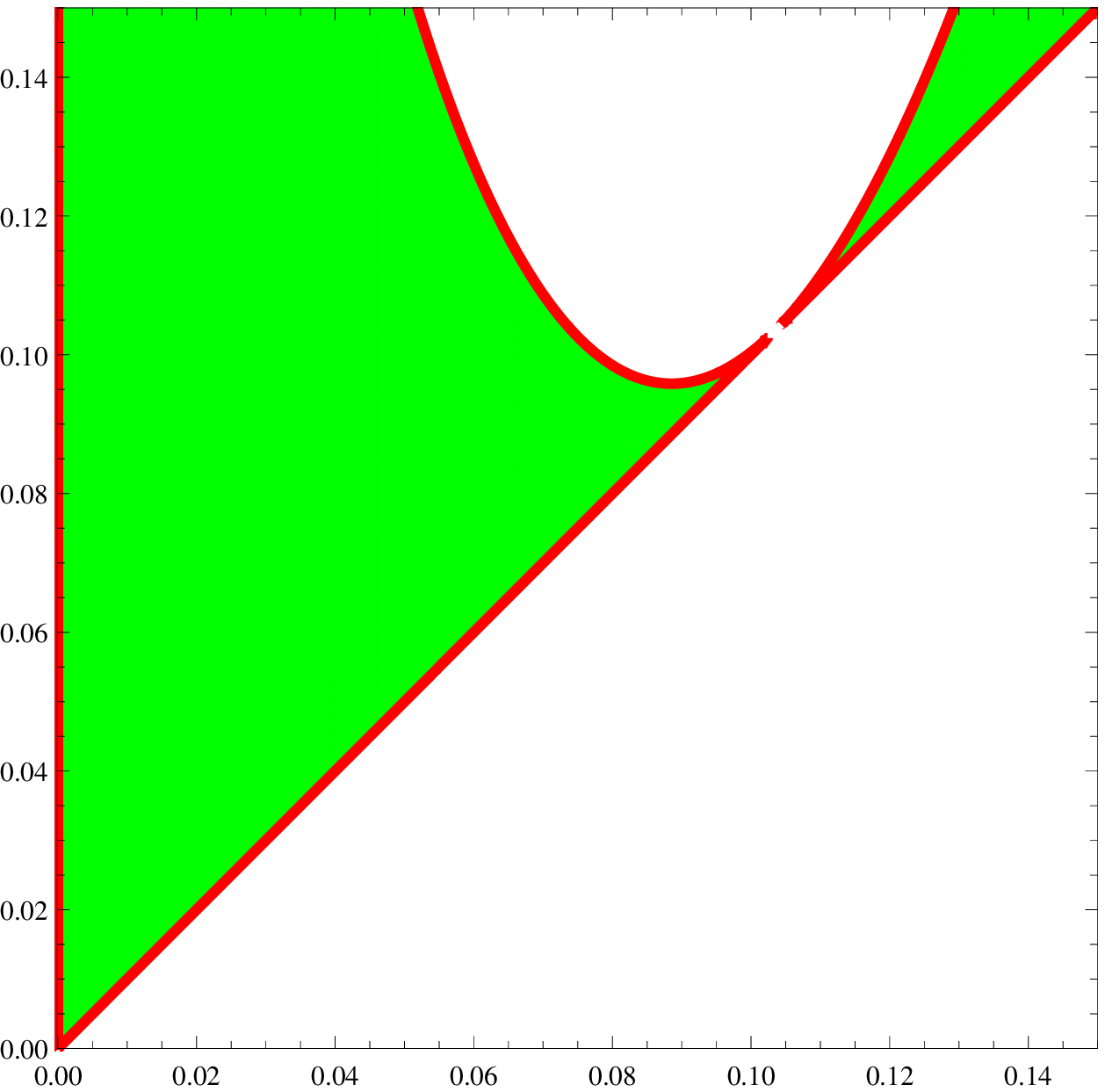}
\linethickness{0.3mm}
\multiput(5.2,64.3)(2,0){48}{\line(1,0){1}}
\put(82,57.5){\scalebox{0.8}{$\eta_2 = \tilde{\delta}$}}
\put(100,0){$\eta_1$}
\put(0,100){$\eta_2$}
\end{overpic}}
\hspace{10pt}
\subfigure[The zero level set of $\mathfrak{R}$]{%
\label{fig:zls}
\begin{overpic}[width=.43\textwidth]{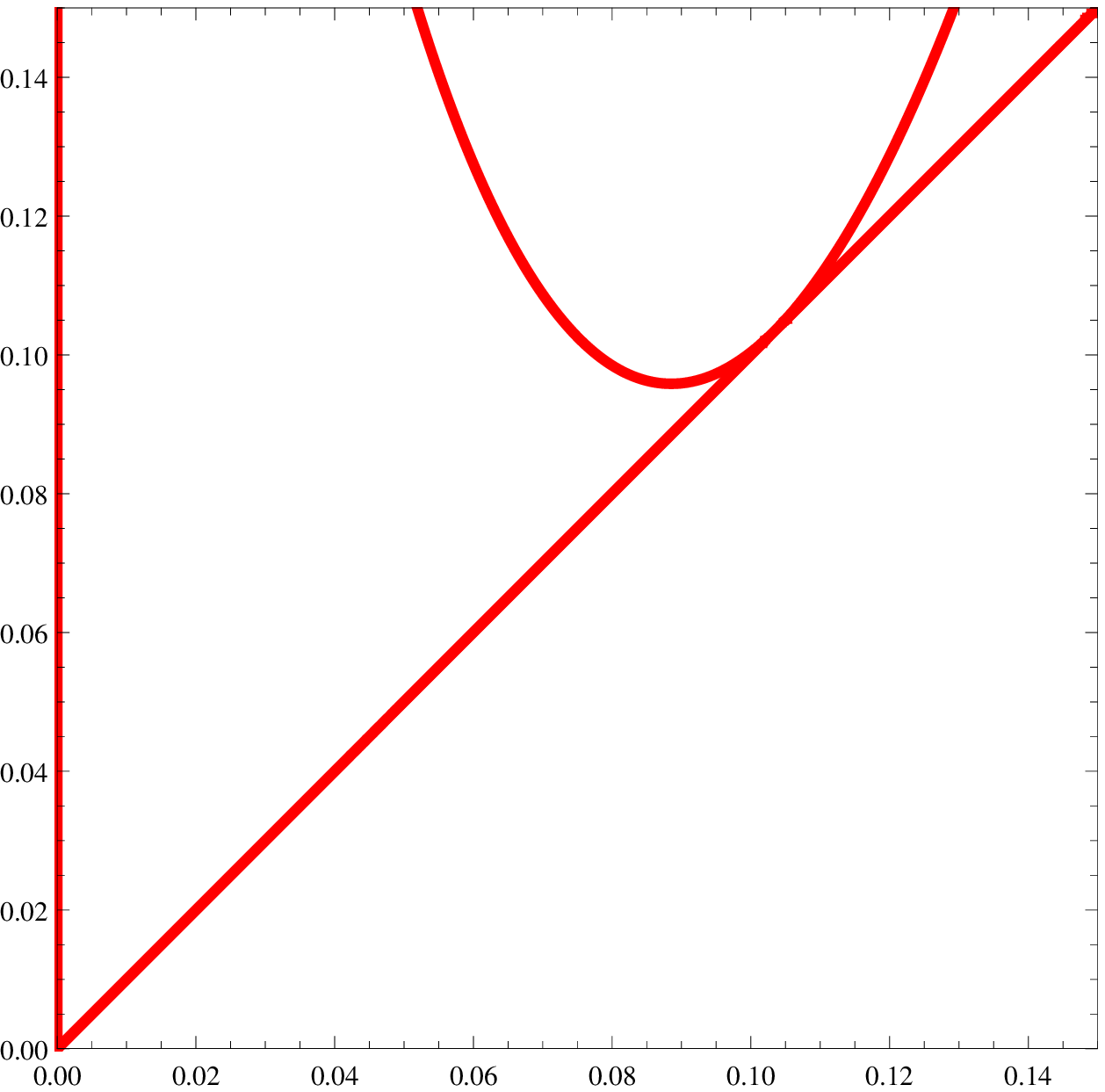}
\linethickness{0.3mm}
\multiput(5.2,64.3)(2,0){48}{\line(1,0){1}}
\put(82,57.5){\scalebox{0.8}{$\eta_2 = \tilde{\delta}$}}
\put(100,0){$\eta_1$}
\put(0,100){$\eta_2$}
\end{overpic}}
\caption{The $x$-axis stands for $\eta_1$, and the $y$-axis stands for
$\eta_2$}
\label{fig:delta}
\end{figure}

In order to determine $\delta_{\max}$, we have to consider two
additional cases: (1) $\eta_1 = 0$, (2) $\eta_1 = \eta_2 > 0$. They
correspond to the straight red lines in Figure \ref{fig:delta}. It can
be argued as below for these cases:
\begin{itemize}
\item For the case $\eta_1 = 0$, if $\eta_2 = 0$, the real
  diagonalizability of $\tilde{\bf M}_1$ has been proven. If $\eta_2 >
  0$, since \eqref{eq:q_M} always holds, we only need to consider
  whether the polynomial $q(\zeta)$ defined in \eqref{eq:q} has
  multiple roots. Figure \ref{fig:res_q} gives the plots of
  $\mathrm{res}(q, q')$ for $\eta_2 \in [0, 0.1]$, where $q'(\zeta) =
  \odd{\zeta} q(\zeta)$.  It is found that if $0 < \eta_2 <
  \tilde{\delta} < 0.1$, then $\mathrm{res}(q, q') > 0$, thus
  $q(\zeta)$ has no multiple roots. Then $\tilde{\bf M}_1$ is real
  diagonalizable.
\begin{figure}[!ht]
\centering
\begin{overpic}[width=.6\textwidth]{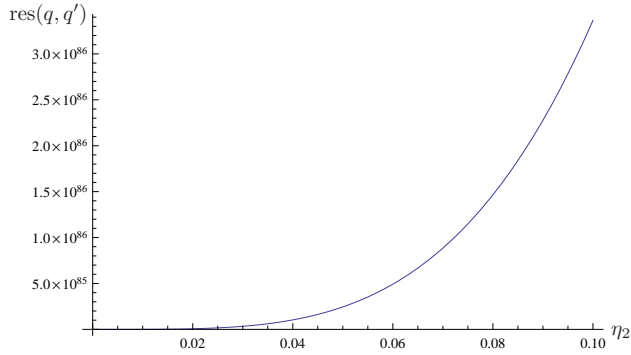}
\put(101,3){\scalebox{0.8}{$\eta_2$}}
\put(-4,58){\scalebox{0.8}{$\mathrm{res}(q,q')$}}
\end{overpic}
\caption{Plots of $\mathrm{res}(q, q')$ in the case of $\eta_1 = 0$}
\label{fig:res_q}
\end{figure}
\item For the case $\eta_1 = \eta_2 > 0$, we have to study the
  multiplicities of the roots of \eqref{eq:p_2}. Denote the polynomial
  \eqref{eq:p_2} by $\tilde{q}$, and let $\tilde{q}'(\zeta) =
  \odd{\zeta} \tilde{q}(\zeta)$. The values of
  $\mathrm{res}(\tilde{q}, \tilde{q}')$ for $\eta_2 \in [0, 0.1]$ are
  given in Figure \ref {fig:res_tilde_q}. It can also be observed that
  when $0 < \eta_1 = \eta_2 < \tilde{\delta} < 0.1$,
  $\mathrm{res}(\tilde{q}, \tilde{q}')$ is greater than zero, which
  results in the real diagonalizability of $\tilde{\bf M}_1$.
\begin{figure}[!ht]
\centering
\begin{overpic}[width=.6\textwidth]{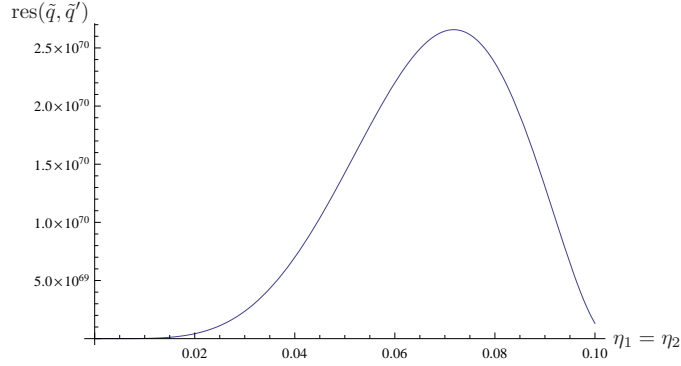}
\put(101,3){\scalebox{0.8}{$\eta_1 = \eta_2$}}
\put(-4,60){\scalebox{0.8}{$\mathrm{res}(\tilde{q},\tilde{q}')$}}
\end{overpic}
\caption{Plots of $\mathrm{res}(\tilde{q}, \tilde{q}')$ in the case of
$\eta_1 = \eta_2$}
\label{fig:res_tilde_q}
\end{figure}
\end{itemize}

As a summary, we claim that if $0 \leqslant \eta_1 \leqslant \eta_2 <
\tilde{\delta}$, the moment system \eqref{eq:new_13m} is hyperbolic.
Thus $\delta_{\max} = \tilde{\delta} \approx 0.095$.

In order to give a more precise description of the size of the
hyperbolicity region, we apply the Chapman-Enskog method to the
modified 13-moment system \eqref{eq:new_13m_mat}. Apply the
transformation $t' = \varepsilon t$ and $\bx' = \varepsilon \bx$ to
\eqref{eq:new_13m_mat}, and then the moment system becomes
\begin{equation} \label{eq:scaled_13m}
\pd{\tilde{\bw}}{t'} +
  \tilde{\bf M}_k(\tilde{\bw})
    \pd{\tilde{\bw}}{x_k'} =
  \frac{1}{\varepsilon} \tilde{\bQ}(\tilde{\bw}).
\end{equation}
For small $\varepsilon$, we formally expand $\tilde{\bw}$ as
\begin{displaymath}
\tilde{\bw} = \tilde{\bw}^{(0)} + \varepsilon \tilde{\bw}^{(1)} +
  \varepsilon^2 \tilde{\bw}^{(2)} + \cdots.
\end{displaymath}
The Chapman-Enskog expansion fixes the leading order term
$\tilde{\bw}^{(0)}$ to be the equilibrium part of $\tilde{\bw}$:
\begin{equation} \label{eq:exp_mnt}
\rho = \rho^{(0)}, \quad \bu = \bu^{(0)}, \quad
{\bf \Theta} = \theta {\bf I} + \varepsilon {\bf \Theta}^{(1)}
  + \varepsilon^2 {\bf \Theta}^{(2)} + \cdots, \quad
\bs = \varepsilon \bs^{(1)} + \varepsilon^2 \bs^{(2)} + \cdots.
\end{equation}
Substituting \eqref{eq:exp_mnt} into \eqref{eq:scaled_13m} and
balancing the zeroth order terms on both sides of
\eqref{eq:scaled_13m}, one may conclude
\begin{displaymath}
\theta_{ij}^{(1)} = -\frac{2\mu}{\rho} \left(
  \pd{v_{(i}}{x_{j)}'} - \frac{1}{3} \pd{v_k}{x_k'}
\right), \quad s_j^{(1)} = -\frac{15\mu}{4\theta} \pd{\theta}{x_i'}.
\end{displaymath}
These are equivalent to the well-known Navier-Stokes and Fourier laws.

According to the expansion \eqref{eq:exp_mnt}, we have
\begin{displaymath}
{\bf \Theta}^{-1} = \theta^{-1} {\bf I} -
  \frac{\varepsilon}{\theta^2} {\bf \Theta}^{(1)} + O(\varepsilon^2),
\end{displaymath}
and thus
\begin{equation} \label{eq:rho_sts}
\rho^{-2} \bs^T {\bf \Theta}^{-1} \bs =
  \varepsilon^2 \rho^{-2} \theta^{-1} |\bs^{(1)}|^2 + O(\varepsilon^3)
  = \frac{225}{16} \varepsilon^2 \mu^2
    \rho^{-2}\theta^{-3} |\nabla_{\bx'} \theta|^2 + O(\varepsilon^3).
\end{equation}
For Maxwellian molecules, the viscosity $\mu$ can be related to the
mean free path $l_{\mathrm{mfp}}$ by
\begin{displaymath}
\mu = \rho~ l_{\mathrm{mfp}} \sqrt{\frac{\pi \theta }{2}} .
\end{displaymath}
Thus, \eqref{eq:rho_sts} is simplified as
\begin{displaymath}
\rho^{-2} \bs^T {\bf \Theta}^{-1} \bs =
  \frac{225\pi}{32} \varepsilon^2 \left(
    \frac{|\nabla_{\bx'} \theta|}{\theta} l_{\mathrm{mfp}}
  \right)^2 + O(\varepsilon^3).
\end{displaymath}
Neglecting the high order terms, we get
\begin{equation} \label{eq:criterion}
\rho^{-2} \bs^T {\bf \Theta}^{-1} \bs \approx
  \frac{225\pi}{32} \varepsilon^2 \left(
    \frac{|\nabla_{\bx'} \theta|}{\theta} l_{\mathrm{mfp}}
  \right)^2 =
  \frac{225\pi}{32} \left(
    \frac{|\nabla_{\bx} \theta|}{\theta} l_{\mathrm{mfp}}
  \right)^2.
\end{equation}
Thus the hyperbolicity condition $\rho^{-2} \bs^T {\bf \Theta}^{-1}
\bs < \delta_{\max}$ is approximately given as
\begin{displaymath}
|\nabla_{\bx} \theta| < C_{\mathrm{hyp}} \theta / l_{\mathrm{mfp}},
  \quad C_{\mathrm{hyp}} = \sqrt{\frac{32\delta_{\max}}{225\pi}}.
\end{displaymath}
Since $\delta_{\max} \approx 0.095$, we have that $C_{\mathrm{hyp}}
\approx 0.065$. Thus the temperature is allow to change around $6.5\%$
of its value in one mean free path in order to ensure the
hyperbolicity. Consider the symmetric plane Couette flow problem. The
Navier-Stokes equations together with the first-order slip boundary
condition is valid only for $l_{\mathrm{mfp}} \leqslant 0.1 L$, where
$L$ is the distance between plates \cite{Karniadakis}. For
$\mathit{Kn} = l_{\mathrm{mfp}} / L$, in order to satisfy the
criterion \eqref{eq:criterion}, the ratio of the temperature in the
middle of the two plates to the temperature on each plate must not
exceed $\mathit{Kn}^{-1} C_{\mathrm{hyp}}$. The numerical results in
\cite{Mieussens2004} show that such a criterion is satisfied even for
very fast plate velocities.


\section{Conclusion} \label{sec:conclusion}

We find that for Grad's 13-moment system, the equilibrium is always on
the boundary of its hyperbolicity region. A modified 13-moment system
is proposed so that the local hyperbolicity around the equilibrium
states can be achieved. The derivation of this new model is almost the
same as the original one, except that the basis functions used in the
expansions of the distribution functions are different. Obviously,
this new model is far away from perfection; most of the classical
criticism on Grad's 13-moment system still applies to this new model.
However, due to the similarity of these two systems, the techniques
developed for Grad's 13-moment system may also apply to this new
model. This modified system enriches the 13-moment family, and some
interesting aspects are found for this new member.


\bibliographystyle{plain}
\bibliography{../article}

\medskip
Received xxxx 20xx; revised xxxx 20xx.
\medskip

\end{document}